\documentclass[format=acmsmall, review=false]{acmart}
\usepackage{acm-ec-21}
\usepackage{booktabs} 
\usepackage[ruled]{algorithm2e} 

\SetAlFnt{\small}
\SetAlCapFnt{\small}
\SetAlCapNameFnt{\small}
\SetAlCapHSkip{0pt}
\IncMargin{-\parindent}
\setcitestyle{acmnumeric}
\newcommand{\SR}{\mathcal{R}}

\newcommand{\Rl}{\mathbb{R}_{\le 0}}
\newcommand{\Rg}{\mathbb{R}_{> 0}}

\newcommand{\R}{\mathbb{R}}

\newcommand{\bra}[1]{\left[#1\right]}
\newcommand{\pa}[1]{\left(#1\right)}
\newcommand{\set}[1]{\left\{ #1 \right\}}

\newcommand{\N}{\mathbb{N}}

\newcommand{\pr}[1]{\Pr\bra{#1}}

\newcommand{\ex}[1]{\mathrm{E}\bra{#1}}

\newcommand{\norm}[1]{\left\lVert#1\right\rVert}

\newcommand{\elsec}{\text{else}}

\newcommand{\X}{\mathcal X}
\newcommand{\dX}{d_{\mathcal X}}

\newcommand{\Ps}{\mathcal P_R}

 \usepackage{wrapfig}
 \usepackage{tikzit}

\tikzstyle{alt}=[fill={rgb,255: red,180; green,214; blue,255}, draw=black, shape=circle]
\tikzstyle{agent}=[fill={rgb,255: red,255; green,116; blue,106}, draw=black, shape=circle]
\tikzstyle{rect}=[fill={rgb,255: red,193; green,193; blue,193}, draw=black, shape=rectangle]

\tikzstyle{grey}=[-, draw={rgb,255: red,152; green,152; blue,152}]
\tikzstyle{row1}=[fill=none, -, draw=red]
\tikzstyle{row2}=[-, draw=green]
\tikzstyle{row3}=[-, draw=blue]
\tikzstyle{dashed}=[-, dashed]
\tikzstyle{thick}=[-]

\title{Utility-Based Communication Requirements for Stable Matching in Large Markets}

\author{Naveen Durvasula}

\begin{abstract}

Results from the communication complexity literature have demonstrated that stable matching requires communication: one cannot find or verify a stable match without having access to essentially all of the ordinal preference information held privately by the agents in the market. Stated differently, these results show that stable matching mechanisms are not robust to even a small number of labeled inaccuracies in the input preferences. In practice, these results indicate that agents must go through the time-intensive process of accurately ranking each and every potential match candidate if they wish for the resulting match to be guaranteedly stable. Thus, in large markets, communication requirements for stable matching may be impractically high. 

A natural question to ask, given this result, is whether some higher-order structure in the market can indicate which large markets have steeper communication requirements. In this paper, we perform such an analysis in a regime where agents have a utility-based notion of preference. We consider a dynamic model where agents only have access to an approximation of their utility that satisfies a universal multiplicative error bound. This bound decays over time by a factor of $D(t)$ as the agents communicate and learn about their preferences, and scales as the size of the market increases by a factor of $H(n)$ as the task of preference learning becomes more difficult. We apply guarantees from the theoretical computer science literature on low-distortion embeddings of finite metric spaces to understand the communication requirements of stable matching in large markets in terms of their structural properties. Our results show that for a broad family of markets, the scale factor $H(n)$ may not grow faster than $n^2\log(n)$ while maintaining a deterministic guarantee on the behavior of stable matching mechanisms in the limit. We also show that a stronger probabilistic guarantee may be made so long as $H(n)$ grows at most logarithmically in the underlying topological complexity of the market. 

\end{abstract}

\begin{document}

\begin{titlepage}

\maketitle

\end{titlepage}

\section{Introduction}\label{sec:intro}
Many real-world matching processes, such as matching students to schools, residents to hospitals, or law clerks to judges, can be classified as \textit{two-sided markets}. In these markets, participants are partitioned into two blocks (abstractly referred to as \textit{men} and \textit{women}), where each member of one block has some well-defined notion of preference over members in the other block. In their seminal paper, Gale and Shapley introduced the deferred acceptance (DA) algorithm to (in $\Theta(n^2)$ time, where $n$ denotes the number of members in each block) find so-called \textit{stable matches} in two-sided markets \cite{gale1962college}. These matches are bijections between the men and women such that no couple mutually prefers to be matched over their assigned partners. Experimental evidence has demonstrated that in a variety of real-world markets (including the three previously mentioned), providing a guarantee of this mathematical property can prevent several common market failures from occurring \cite{roth2008deferred, abdulkadirouglu2005new, avery2001market}. This discovery has led to the wide-spread adoption of stable-matching mechanisms \cite{roth2008deferred}.

The DA algorithm may be regarded as ``efficient'' insofar as it runs in polynomial time with respect to the number of agents in the market. However, it's not immediately obvious whether this algorithmic efficiency translates into an efficient clearing process in real-world markets. In markets such as the school choice, residency, and law clerk markets mentioned above, the time it takes for the mechanism to clear is limited primarily by the cost of interaction between the agents rather than the computational complexity of the matching algorithm. In the process of forming preferences, agents may have to tour schools, visit hospitals, or complete time-intensive interviews. Indeed, in larger markets, the cost of ranking every agent on the other side of the market may be impractically high. Yet, the classical DA algorithm requires such complete preferences as input. A natural question to ask, therefore, is whether a stable matching may be found (or at least verified) in a more communication-efficient way.

This problem has been historically studied through the lens of communication complexity, where the minimum number of queries necessary to find or verify a stable matching is analyzed given a specified communication model for the agents in the market. Each agent is assumed to privately hold a strict preference ranking across the agents on the other side of the market. One of the first such analyses is due to Ng and Hirschberg, who showed that $\Theta(n^2)$ queries are required to find or verify a stable matching in a model where a central party may query for (a) a man $m$'s ranking of a woman $w$, and (b) the woman $w$ that man $m$ places at rank $k$ (or vice-versa) \cite{ng1990lower}. Chou and Lu extend this analysis and show that even if the central party were able to separately query for each of the bits in the integer responses for (a) and (b), $\Theta(n^2 \log n)$ Boolean queries are required to find a stable matching \cite{chou2010communication}. Further analyses by Segal \cite{segal2007communication} and Gonczarowski et al.\cite{gonczarowski2014stable} show that under \textit{any} Boolean framework, $\Omega(n^2)$ queries are required to find or verify an (approximate) stable matching, even when one is allowed to use nondeterministic or randomized protocols. Table \ref{tab:results} gives an overview of these results. 

\begin{table}[h!] 
\centering
\resizebox{0.8\columnwidth}{!}{\begin{minipage}{\textwidth}\begin{center}

    \begin{tabular}{|c||c|c|c|}
        \hline
         & Query Model & Protocol Type & Bits Required  \\
        \hline
        \hline
        \cite{ng1990lower} & Queries for (a) and (b)  & Deterministic & $\Theta(n^2)$ \\
        \hline
        \cite{segal2007communication} & Any Boolean queries & (Non)deterministic & $\Omega(n^2)$ \\
        \hline
        \cite{chou2010communication} & Queries with $\log n$ bit responses & Approximate & $\Omega(n^2 \log n)$ \\
        \hline
        \cite{gonczarowski2014stable} & Any Boolean queries & Randomized or Approximate & $\Omega(n^2)$ \\
        \hline
        \cite{ashlagi2017communication} & \begin{tabular}{@{}c@{}} Any Boolean queries, and\\agents have some prior knowledge\end{tabular} & Randomized & $\tilde O(n^{3/2})$ \\
        \hline 
    \end{tabular}\end{center}\caption{\label{tab:results}A summary of the results on the communication complexity of stable matching. The notation $\tilde O$ hides $\log$ factors in $n$. A randomized protocol must give a stable match with probability at least $\frac23$. An approximate protocol must give a match that \textit{almost} stable in some well-defined sense. Different notions of approximate stability are used in \cite{gonczarowski2014stable} and \cite{chou2010communication}, so these results are not directly comparable. The terms ``nondeterministic'' and ``randomized'' have different meanings in the communication complexity literature. }\end{minipage}
    }
    \vspace{-1.5em}
\end{table}

The results in Table \ref{tab:results} demonstrate that essentially no asymptotic improvement may be made to the DA algorithm to reduce the communication complexity: all but a constant fraction of the information contained in the full preferences must be communicated in order to find or verify a stable matching in the worst case. Thus, they provide theoretical evidence that communication requirements for stable matching in large markets may be impractically high. 

These results, however. do not allow us to understand the communication requirements for stable matching in terms of some underlying structure in the market. We may not, for example, understand why communication requirements might differ in two markets that are of the same size through the communication complexity model. Further, in the economic interpretation of the communication complexity as a proxy for the amount of time agents must spend communicating in the market, one must assume assume a discrete model where agents are queried in individual time steps. It is not immediately clear that the impossibility results given in Table \ref{tab:results} carry over to the more realistic model where agents are continuously communicating. 

In this paper, we fill in these gaps in the literature by allowing agents to express strength of preference by means of a utility function. We develop a framework to study a continuous analogue of the communication complexity. We consider a dynamic model in which agents begin with a vague understanding of their true utility. As they communicate in the market for a longer amount of time, they learn about their preferences and recover successively more accurate approximations to their true utility. We consider both deterministic and probabilistic perturbations. More formally, we give each agent access to some perturbed utility function that approximates their true utility function up to a multiplicative error bound that depends on (i) the amount of time $t$ the agents have been communicating, and (ii) the size of the market $n$. As the amount of time $t$ the agents spend communicating increases, the error bound decays by a factor of $D(t)$ that increases monotonically over time. However, as the number of agents in the market increases, making a universal error guarantee on the full preferences becomes increasingly difficult. We model this added difficulty by allowing the error bound to scale proportionally to a nondecreasing \textit{hardness} function $H(n)$ that indicates how error increases as size of the market increases. In practical terms, the decay $D(t)$ indicates how fast agents in the market learn preferences, and the hardness $H(n)$ indicates the rate at which preference learning becomes more difficult as the size of the market increases. In our probabilistic model, we similarly consider random approximations to the agents' utility functions that abide by the multiplicative error bound in expectation. 

We then consider the \textit{communication requirement} for the stable matching market: the infimal amount of time that the agents must communicate in order to provide a guarantee that the male- and female-optimal stable matches under the perturbed utility functions are equal to those under the true utility functions. We study the conditions for communication requirements to stay finite in the limit as the size of the market tends to infinity. Indeed, if the communication requirement tends to infinity, this indicates that the stable matching mechanisms will not function properly in large markets: even if agents submit preference lists to the matching mechanism according to their current approximation of their utility, the output of the mechanism may not be stable according to the agents' true preferences! As it turns out, the conditions for communication requirements to stay finite in the limit depend only on the hardness function $H(n)$ -- we call such $H(n)$ that satisfy this property \textit{admissible}. The communication requirement, which serves as an analogue to the communication complexity, is closely related to the notion of \textit{robustness}: the supremal amount by which the true utility may be perturbed multiplicatively while maintaining a guarantee that with positive probability, the induced stable matches remain the same. Indeed, the robustness may also be used to understand the class of admissible $H(n)$ and is a central object of study in this paper.

The model described above resolves the discreteness issue that communication complexity suffers from. Next, we provide a condition on the agents' utilities that allows us to understand the market in terms of some additional underlying structure. This linear constraint, which we call \textit{polarity}, is likely to hold in markets with conventionally ``polarized'' preferences. Informally, we assert that if an agent strongly prefers one match candidate to another, then every other agent must strongly dislike at least one of the two candidates in question. We show that in polarized markets (and \textit{only} polarized markets), there exists an mapping of the agents into a finite metric space such that distance is equal to negative utility. Thus, these markets can be thought of as a generalization of well-studied spatial models of preference in the economics and political science literature \cite{hotelling1990stability, eguia2011foundations, bogomolnaia2007euclidean, davis1972social, enelow1984spatial, anshelevich2017randomized} where agents' preferences are given by Euclidean distances. We call any metric space satisfying this property a \textit{generating metric space} for the polarized market. In both our deterministic and probabilistic models, we prove upper bounds on the robustness of stable matching to deterministic and probabilistic perturbations. Our bounds are given in terms of (i) the size of a generating metric space (ii) the topological genus of a generating metric space, and (iii) the size of the market. Our bounds are logarithmic in the size and genus of the generating metric space, and roughly quadratic in the size of the market. These upper bounds correspond to lower bounds on the communication requirement and restrictions on the class of admissible hardness functions $H(n)$. 

Our results and methods, through the lens of matching theory, provide what we believe is a rich connection between spatial models of preference and sketching/embedding algorithms from the theoretical computer science literature \cite{johnson1984extensions, bourgain1985lipschitz, sidiropoulos2010optimal, fakcharoenphol2004tight}. Indeed, low-distortion embedding algorithms are the key tool that we use to generate our robustness bounds. Informally, these algorithms show that one may transform any arbitrary finite metric space into a restricted class of metric spaces without significantly altering the distances. These results have historically been used to speed up the running time of various linear algebra applications and combinatorial algorithms\footnote{We direct the reader to \cite{indyk2004low} for a survey of these techniques and their applications} \cite{cormode2005improved, drineas2001fast, erickson2012combinatorial}. However, we show that they also have an economic interpretation: by relating two classes of metric spaces, these results also relate classes of markets under a spatial model of preference. Thus, we use these results to generalize bounds we show in one class into bounds on the other class.

In Section \ref{sec:prelims} we define our model more formally, starting with the relevant definitions for one-sided markets and working our way towards the corresponding definitions for two-sided matching markets. We give both a deterministic and probabilistic model for the communication process. We also show how the aforementioned notions of the communication requirement, robustness, and admissibility are related. In Section \ref{sec:general}, we give more intuition for these concepts by deriving an explicit formula for the deterministic robustness, and proving a relationship between our deterministic and probabilistic models. In Section \ref{sec:polarization}, we define the polarity condition, and show that polarized markets (and only polarized markets) may be associated with a metric space in such a way that distance equals negative utility. We also give an economic interpretations for quantities pertaining to a generating metric space, such as its size and genus. In Section \ref{sec:robustness}, we show that if agents have Euclidean preferences, then the only admissible hardness function $H(n)$ is a constant function. That is, in Euclidean markets, preference learning may not get more difficult as the size of the market increases if we wish to obtain a guarantee on the behavior of stable matching in large markets. We then make use of Bourgain's embedding of arbitrary metric spaces into $\ell_2$ \cite{bourgain1985lipschitz} to generalize this statement to all polarized markets. In Section \ref{sec:probounds}, we make use of Sidiropolous's probabilistic embedding of arbitrary metric spaces into genus zero metric spaces \cite{sidiropoulos2010optimal} to make similar bounds in our probabilistic model. Finally, in Section \ref{sec:conclusion} we give an interpretation of our results and suggest related open problems.

\section{Preliminaries}\label{sec:prelims}

\subsection{One-Sided Markets}
In this section, we define the basic terminology that we will use for the remainder of the paper. We first introduce the relevant terminology for one-sided markets. We consider $n$ \textit{agents} who have preferences over $n$ \textit{alternatives}. An \textit{assignment} $\mu: [n] \leftrightarrow [n]$ is a bijection between the agents and alternatives. We denote by $\SR^n$ the set of all strict preference profiles (linear orderings) $n$ agents may have over $n$ alternatives. For $R = (R_a)_{a \in [n]} \in \SR^n$, we write $xR_ax'$ to indicate that agent $a$ prefers $x$ to $x'$.

We similarly allow agents to have utility-based descriptions of preference. In this paper, we assume the existence of a universal constant that upper bounds the utility any agent may receive from the assignment to an alternative. We define an agent's utility for an assignment to a given alternative as a (nonpositive) difference from this universal constant. Formally, an \textit{$n$-utility profile} $u: [n] \times [n] \mapsto \Rl$ associates each agent-alternative pair $(a,x) \in [n] \times [n]$ with a nonpositive number representing the utility $a$ receives under an assignment to $x$. If $u$ is \textit{strict} (i.e. if $u(a,x) = u(a,x') \iff x = x'$), then we refer to $\SR(u) \in \SR^n$ as the unique preference profile satisfying $x\SR(u)_ax' \iff u(a,x) > u(a,x')$.

\begin{definition}[Market Profile]
A market can be characterized by the collection of utility profiles that may possibly arise. Formally, we let an \textit{$n$-market profile} $U: \SR^n \to \Rl^{[n] \times [n]}$ be a map that gives a utility-based description of the agents' preferences given the ordinal description of their preferences. We enforce that for any $R \in \SR^n$, the utility profile $U[R]: [n] \times [n] \mapsto \Rl$ satisfies the property $\SR(U[R]) = R$.
\end{definition}

In our communication model, agents operate under a perturbed utility profile that becomes increasingly more accurate as more communication takes place. We now propose a model for such perturbations. We give two formulations: one deterministic and one probabilistic. In our deterministic formulation, the agents' original utility profile $u$ is replaced by a profile $\delta u$ that has been perturbed by a bounded multiplicative distortion $\delta$. We make the assumption that such distortions may never cause an agent to overestimate the utility they may receive from any assignment to an alternative. 

\begin{definition}[$C$-Perturbation]\label{def:cpert}
A $C$-perturbation is a map $\delta: [n] \times [n] \to \Rg$ such that for any $a, x \in [n]$
\[1 \le \delta(a,x) \le C\]
It then follows that for any utility profile $u: [n] \times [n] \to \Rl$ and any $C$-perturbation $\delta$, 
\[Cu(a,x) \le (\delta u)(a,x) \le u(a,x)\]
for any $a,x \in [n]$.
\end{definition}

In the probabilistic formulation, the agents' original profile $u$ is replaced by a random utility profile $\delta u$ that has been perturbed by a random multiplicative distortion $\delta$ that is bounded in expectation. 

\begin{definition}[Probabilistic $C$-Perturbation]\label{def:probcpert}
A probabilistic $C$-perturbation is a random map $\delta: [n] \times [n] \to \Rg$ such that the following two conditions hold for any $a, x \in [n]$
\[\delta(a,x) \ge 1 \qquad \ex{\delta(a,x)} \le C\]
It then follows that for any utility profile $u$ and any probabilistic $C$-perturbation $\delta$,
\[C u(a,x) \le \ex{(\delta u)(a,x)} \le u(a,x)\]
The random variables $\delta(a,x)$ may be arbitrarily distributed, and need not be independent with each other. 
\end{definition}

Note that any deterministic or probabilistic $C$-perturbation of a utility profile $\delta u$ must fix a utility of zero -- the highest utility possible in our framework -- and no such perturbation may change a nonzero value to zero. Thus, in both models, an agent $a$ receives the theoretical maximum utility under an assignment to $x$ if and only if $a$ is aware of this fact. Conversely, if $u'$ is any deterministic/random $n$-utility profile satisfying the properties $u'(a,x) = 0 \iff u(a,x) = 0$ and $u'(a,x) \le u(a,x)$, then it may be written as $\delta u$ for some deterministic/probabilistic $C$-perturbation $\delta$. 

\subsection{Two-Sided Markets}
We now use this terminology to define and characterize two-sided matching markets. We consider collections of $n$ \textit{men} and $n$ \textit{women} who have preferences over each other. Given preference profiles $R^M, R^W \in \SR^n$ defining the joint preferences each side of the market has over the other, a man $m$ and a woman $w$ form a \textit{blocking pair} in an assignment $\mu: [n] \leftrightarrow [n]$ if $wR^M_m\mu(m)$ and $mR^W_w\mu^{-1}(w)$. An assignment $\mu$ is \textit{stable} if no such blocking pairs exist. We let $\Phi(R^M,R^W) := (\mu_M, \mu_W)$ denote the \textit{deferred acceptance operator}, which takes as input preference profiles $R^M$ and $R^W$ and returns the male-optimal and female-optimal stable assignments $\mu_M$ and $\mu_W$ as would be returned by running the man-proposing and woman-proposing deferred acceptance algorithm (see \cite{gale1962college} for an exposition).

\begin{definition}[Matching Market]
An $n$-matching market $(U_M, U_W)$ consists of two market profiles $U_M, U_W : \SR^n \to \Rl^{[n] \times [n]}$, and gives a utility-based description of the agents' preferences on each side of the market given the ordinal description of their preferences. If the joint preferences of the $n$ men are given by a preference profile $R^M$, we denote by $U_M[R^M]$ their utility profile in this matching market. We similarly denote by $U_W[R^W]$ the utility profile of the $n$ women given that they have joint preference $R^W$. 
\end{definition}

We use this framework to characterize the robustness of a matching market to inaccuracies in preference learning -- the degree to which its utility profiles may be perturbed while still maintaining the possibility that the resulting pair of stable assignments given by $\Phi$ remains unchanged. As before, we give two characterizations based on Definitions \ref{def:cpert} and \ref{def:probcpert}.

\begin{definition}[Robustness]\label{def:rob}
An $n$-matching market $(U_M, U_W)$ is $C$-robust if for all men's and women's preference profiles $R^M, R^W \in \SR^n$ and $C$-perturbations $\delta_M, \delta_W$, \[\Phi(R^M,R^W) = \Phi(\SR(\delta_M U_M[R^M]), \SR(\delta_W U_W[R^W]))\] We define the \textit{robustness} $\xi_{U_M, U_W}$ of the matching market $(U_M, U_W)$ as 
\[\xi_{U_M, U_W} := \sup\pa{\set{C \ge 1 \mid (U_M, U_W)\text{ is $C$-robust}}}\]
\end{definition}

In other words, the robustness $\xi_{U_M, U_W}$ is the largest real number such that for every pair of preference profiles $R^M$ and $R^W$, we may perturb $U_M[R^M]$ and $U_W[R^W]$ by up to $\xi_{U_M, U_W}$ without changing the male-optimal and female-optimal stable assignments.

\begin{definition}[Probabilistic Robustness]\label{def:prob}
An $n$-matching market $(U_M, U_W)$ is probabilistically $C$-robust if for any joint distribution over preference profiles $R^M, R^W \in \SR^n$ and probabilistic $C$-perturbations $\delta_M, \delta_W$,
\[\pr{\Phi(R^M,R^W) = \Phi(\SR(\delta_M U_M[R^M]), \SR(\delta_W U_W[R^W]))} > 0\]
We define the \textit{probabilistic robustness} $\xi^P_{U_M, U_W}$ of the matching market $(U_M, U_W)$ as 
\[\xi^P_{U_M, U_W} := \sup\pa{\set{C \ge 1 \mid (U_M, U_W) \text{ is probabilistically $C$-robust}}}\]
\end{definition}

Similarly, the probabilistic robustness $\xi^P_{U_M, U_W}$ is the largest real number such that for any random initialization of the market, we may perturb $U_M[R^M]$ and $U_W[R^W]$ by up to $\xi_{U_M, U_W}^P$ in expectation and it will remain possible for the male-optimal and female-optimal stable assignments to stay unchanged. 

\subsection{Modeling the Preference Learning Process}\label{subsec:model}
We consider a dynamic model for the preference learning process where the agents at time $t$ operate in an $n$-matching market $(\delta_M U_M, \delta_W U_W)$ where $\delta_M$ and $\delta_W$ are each $C(n,t)$-perturbations. In other words, we assert that in a matching market of size $n$, agents recover their true utility up to a multiplicative factor of $C(n,t)$ within $t$ time. We further assume that $C(n,t)$ takes a particular form:
\begin{equation}
C(n,t) := \frac{H(n)}{D(t)} \qquad \qquad {H: \N \to \R_{> 0} \atop D: \R \to \R_{> 0}}
\end{equation}
The term $H(n)$ is a non-decreasing function that indicates how the hardness of preference learning increases with the size of the matching market. The term $D(t)$ is a monotonically increasing decay factor that indicates how the agents learn better approximations to their true utility over time. We further assume that $D$ is continuous and that $\lim_{t \to \infty}D(t) = \infty$. That is, as time tends to infinity, the agents recover their true utility up to an arbitrarily small multiplicative error. 

We now define the \textit{communication requirement} of an $n$-matching market $(U_M, U_W)$: our analogue to the communication complexity when preferences are given by a real-valued utility model. 
\begin{definition}[Communication Requirement]
The communication requirement $T_{(U_M, U_W)}$ of an $n$-matching market $(U_M, U_W)$ is given by the infimum
\[T_{(U_M, U_W)} := \inf\pa{\set{t > 0 \mid (U_M, U_W) \text{ is $C(n,t)$-robust}}}\]
We similarly define the probabilistic communication requirement $T^P_{(U_M, U_W)}$ as the infimum
\[T_{(U_M, U_W)}^P := \inf\pa{\set{t > 0 \mid (U_M, U_W) \text{ is probabilistically $C(n,t)$-robust}}}\]
\end{definition}
The communication requirement is the earliest time at which we may guarantee that the male- and female-optimal stable assignments under the perturbed utilities are equal to those under the true utility. In the probabilistic case, we only require that such a guarantee can be made with positive probability. Conversely, if such stable assignments are computed prior to the communication requirement, we may guarantee that these assignments are almost surely \textit{not} stable with respect to the true preferences of the agents. Compare this with the communication complexity, which is given by the fewest number times one needs to query the agents in the market in order to make a similar guarantee. By associating a time cost to each query, one may translate some analogue of the communication complexity results given in Table \ref{tab:results} into bounds on a communication requirement. 

The communication requirement is closely related to the robustness quantities given in Definitions \ref{def:rob} and \ref{def:prob}. The following proposition makes this relationship concrete
\begin{proposition}\label{prop:rob}
Let $(U_M, U_W)$ be an $n$-matching market. Then, the following hold:
\[T_{(U_M, U_W)} = D^{-1}\pa{\frac{H(n)}{\xi_{U_M, U_W}}} \qquad \qquad T^P_{(U_M, U_W)} = D^{-1}\pa{\frac{H(n)}{\xi^P_{U_M, U_W}}}\]
\end{proposition}
The proof follows by noting that at the communication requirement, the value of $C(n,T_{(U_M, U_W)})$ is equal to the robustness, due to the monotonicity of $D$. As $D$ is monotonic and continuous, it is also invertible, whence we arrive at the desired result. The argument is symmetric for the probabilistic case.

Proposition \ref{prop:rob} demonstrates that for any finite $n$, the (probabilistic) communication requirement of any $n$-matching market $(U_M, U_W)$ must be finite. This follows from the fact that all matching markets are (probabilistically) $1$-robust. However, the behavior of the communication requirement in the limit depends on the hardness function $H(n)$, and has tangible economic significance. Indeed, if communication requirements can become arbitrarily large in large markets, then for such markets, the assignments given by the DA algorithm will not be stable! We call hardness functions that give finite communication requirements in the limit \textit{admissible}.

\begin{definition}[Admissibility]
Let $\set{(U_M, U_W)_n}_{n=1}^\infty$ be a collection of matching markets of increasing size, where $(U_M, U_W)_n$ is an $n$-matching market. A hardness function $H(n)$ is \textit{admissible}/\textit{probabilistically admissible} for this collection if
\[\lim_{n \to \infty} T_{(U_M, U_W)_n} < \infty \qquad \qquad \lim_{n \to \infty} T^P_{(U_M, U_W)_n} < \infty\]

Just as Proposition \ref{prop:rob} shows that communication requirements may be understood in terms of the robustness, we show that admissibility can also be understood in this way
\begin{proposition}\label{prop:adm}
Suppose that $H(n)$ is admissible/probabilistically admissible for the collection of matching markets $\set{(U_M, U_W)_n}_{n=1}^\infty$. Then, $H(n) = O\pa{\xi_{(U_M, U_W)_n}}$/$H(n) = O\pa{\xi^P_{(U_M, U_W)_n}}$.
\end{proposition}
\begin{proof}
We show that $H(n) = O\pa{\xi_{(U_M, U_W)_n}}$. The argument for probabilistic case is symmetric. By Proposition \ref{prop:rob}, we have that
\[\lim_{n \to \infty} T_{(U_M, U_W)_n} = \lim_{n \to \infty} D^{-1}\pa{\frac{H(n)}{\xi_{(U_M, U_W)_n}}}\]
As $D$ is monotonically increasing, so too is $D^{-1}$. Further, as $\lim_{t \to \infty} D(t) = \infty$, we must also have that $\lim_{c \to \infty}D^{-1}(c) = \infty$. It then follows that the above limit is finite if
\[\lim_{n \to \infty} \frac{H(n)}{\xi_{(U_M, U_W)_n}} < \infty \implies H(n) = O\pa{\xi_{(U_M, U_W)_n}}\]
\end{proof}

Thus, for the remainder of this paper, we focus our attention on upper bounding the robustness quantities $\xi_{U_M, U_W}$ and $\xi^P_{U_M, U_W}$. Such bounds, by means of Propositions \ref{prop:rob} and \ref{prop:adm}, immediately induce corresponding lower bounds on the communication requirement and describe the admissible class of hardness functions for a given countable collection of matching markets. 
\end{definition}

\section{Warming Up: Understanding the Basic Properties of Robustness}\label{sec:general}

In this section, we give an explicit formula for the robustness $\xi_{U_M, U_W}$ of an $n$-matching market and prove a relationship between the robustness $\xi_{U_M, U_W}$ and the probabilistic robustenss $\xi^P_{U_M, U_W}$. Our arguments make use of the following fact about stable assignments.

\begin{lemma}\label{lem:ne}
Let $R, R' \in \SR^n$. If $R \ne R'$, then there exist some $R^M,R^W \in \SR^n$ such that $\Phi(R, R^W) \ne \Phi(R', R^W)$ and similarly $\Phi(R^M,R) \ne \Phi(R^M,R')$.
\end{lemma}
\begin{proof}
We show the existence of an $R^W \in \SR^n$ such that $\Phi(R, R^W) \ne \Phi(R', R^W)$. The argument is symmetric in the other case. As $R \ne R'$, there exists some man $m_1$ and women $w_1$ and $w_2$ such that $w_1R_{m_1}w_2$, but $w_2R'_{m_1}w_1$. Let $R^W$ be any preference profile where
\[R^W: = \begin{array}{|c|c|c|c|c|}
\hline
w_1 & w_2 & w_3 &\hdots & w_n \\
\hline
\hline
m_1 & m_1 & m_3 & \hdots & m_n \\
m_2 & m_2 & * & \hdots & * \\
* & * & * & \hdots & * \\
\vdots & \vdots & \vdots & \vdots & \vdots \\
\hline
\end{array}\]
We show that the female-optimal assignment $\mu_W$ under the preference profiles $R$ and $R^W$ differs from the female-optimal assignment $\mu_W'$ under the preference profiles $R'$ and $R^W$. Note that by construction, all women except $w_1$ and $w_2$ are matched to their top choice man in the first iteration of the deferred acceptance algorithm. Since $w_1R_{m_1}w_2$, in the assignment $\mu_W$, $m_1$ is matched to $w_1$ and $m_2$ is matched to $w_2$. However, as $w_2R'_{m_1}w_1$, in the assignment $\mu_W'$, $m_1$ is matched to $w_2$ and $m_2$ is matched to $w_1$. Thus, $\Phi(R, R^W) \ne \Phi(R', R^W)$.
\end{proof}

We now use this fact to show that an $n$-matching market is $C$-robust if and only if the strength of any given agent's preference for a match candidate is exponential in that candidate's ordinal ranking. In other words, the market is $C$-robust iff it is impossible to perturb the ordinal preferences of any agent by distorting their utility multiplicatively by a factor of $C$. This result is similar in nature to the negative results in the communication complexity literature: it demonstrates that one may not make a guarantee on stable assignments without making a guarantee on the full preferences of the agents.

\begin{theorem}\label{thm:c}
An $n$-matching market $(U_M, U_W)$ is $C$-robust if and only if for all preference profiles $R \in \SR^n$ and all $m.w,w' \in [n]$, $wR_mw' \implies CU_M[R](m,w) > U_M[R](m,w')$ and similarly for $U_W[R]$.
\end{theorem}

\begin{proof}
We first prove the ``if'' direction. Suppose that for all $R \in \SR^n$ and $m,w,w' \in [n]$, $wR_mw' \implies CU_M[R](m,w) > U_M[R](m,w')$. As a $C$-perturbation $\delta$ may only distort utility profiles up to a multiplicative factor of $C$, the ordinal data contained in $\delta U_M[R]$ must be preserved for any $R$. That is, $\SR(\delta U_M[R]) = R$ for any $C$-perturbation $\delta$ and $R \in \SR^n$. By symmetry, it follows that $\Phi(R^M,R^W) = \Phi(\SR(\delta_M U_M[R^M]), \SR(\delta_W U_W[R^W]))$ for any $C$-perturbations $\delta_M$ and $\delta_W$ and preference profiles $R^M, R^W \in \SR^n$, so $(U_M, U_W)$ is $C$-robust. 

We now show the ``only if'' direction by contradiction. Without loss of generality, suppose that for some preference profile $R \in \SR^n$, there exists some man $m^* \in [n]$ and women $w_1, w_2 \in [n]$ such that $w_1R_{m^*}w_2$, but $CU_M[R](m^*,w_1) < U_M[R](m^*,w_2)$. Let $\delta$ be the $C$-perturbation given by
\[\delta(m,w) = \begin{cases}C & \text{if } m = m^*, w = w_1 \\ 1 & \text{else}\end{cases}\]
Observe that $R \ne \SR(\delta U_M[R])$, since $w_2\SR(\delta U_M[R])_{m^*}w_1$. Thus, by Lemma \ref{lem:ne}, there exists some $R^W \in \SR^n$ such that $\Phi(R, R^W) \ne \Phi(\SR(\delta U_M[R]), \SR(U_W[R^W]))$. It follows that $(U_M, U_W)$ is not $C$-robust.
\end{proof}

Using Theorem \ref{thm:c}, we may directly solve for the robustness $\xi_{U_M, U_W}$ of an $n$-matching market.
\begin{equation}\label{eq:robust}
\xi_{U_M, U_W} = \min_{R \in \SR^n} \pa{ \min_{m,w,w' \in [n] \mid wR_mw'}\pa{\frac{U_M[R](m,w')}{U_M[R](m,w)}},  \min_{w,m,m' \in [n] \mid mR_wm'}\pa{ \frac{U_W[R](m',w)}{U_W[R](m,w)}}}
\end{equation}
By applying Proposition \ref{prop:rob}, we obtain an equation for the communication requirement. This equation partially confirms the intuition that large markets are less robust, as the outer minimum is being taken over the set $\SR^n$ whose size increases factorially in $n$. It is, however, difficult to extract more information out of the equation as we have placed no further assumptions on the matching market. In the next section, we show that a single linear constraint on the agents' utilities allows us to extract far more information. 

A similar formula for the probabilistic robustness of an $n$-matching market is far more difficult to derive. The deterministic case was easier to analyze due to the fact that each entry of a $C$-perturbation $\delta$ may be independently set between $1$ and $C$. In the probabilistic case, the entries of the perturbation may depend on each other. Further, some entries may far exceed $C$, so long as they are bounded in expectation. 

We can, however, understand the probabilistic robustness as it relates to the (deterministic) robustness. Clearly, we must have that $\xi_{U_M, U_W}^P \le \xi_{U_M, U_W}$ for any $n$-matching market $(U_M, U_W)$, since any $C$-perturbation is also a probabilistic $C$-perturbation. A natural question to ask, therefore, is whether one may show a similar upper bound on the probabilistic robustness in terms of the robustness. We resolve this question in the positive, and show that $\xi_{U_M, U_W} = O(n^2)\xi_{U_M, U_W}^P$ is a tight bound. 
\begin{theorem}\label{thm:ub}
If $(U_M, U_W)$ is a $\pa{2n(n-1)(C-1) + 1}$-robust $n$-matching market, then it is probabilistically $C$-robust, but not necessarily probabilistically $(1+\epsilon)C$-robust for any $\epsilon > 0$.
\end{theorem}

See Appendix \ref{proof:ub} for the full proof. By applying Proposition \ref{prop:adm} to the above result, we obtain a corollary of economic significance
\begin{corollary}
For any collection of matching markets $\set{(U_M, U_W)_n}_{n=1}^\infty$,
\begin{equation}
    H^P(n) \le H(n) \le O(n^2)H^P(n)
\end{equation}
where $H^P$ is probabilistically admissible over the collection, and $H$ is admissible over the collection.
\end{corollary}
This result may be interpreted as follows. Suppose that obtaining a universally better approximation of one's preferences becomes more difficult as the size of the matching market $n$ grows by a factor of $H^P(n)$. If, in the limit, one may approximate their true preferences sufficiently well in finite time up to a probabilistic guarantee under these conditions, then one may similarly approximate their true preferences sufficiently well up to a deterministic guarantee in conditions $H(n)$ that are at most quadratically \textit{harder} in the size of the market. Conversely, if one starts with a deterministic guarantee under some conditions $H(n)$, then a probabilistic guarantee may only be made in conditions $H^P(n)$ that are \textit{easier}. Thus, the probabilistic guarantee is stronger than the deterministic guarantee, but not unboundedly so.

\section{Polarization and Matching Markets}\label{sec:polarization}

In this section, we introduce a linear constraint on matching markets which we call \textit{polarity}. We then show that in markets satisfying this constraint (and only markets satisfying this constraint), one may give a geometric interpretation for the agents' utilities. In the later sections of this paper, we show upper bounds on $\xi_{U_M, U_W}$ and $\xi_{U_M, U_W}^P$ in terms of the underlying geometric structure. Indeed, our bounds on the probabilistic robustness are given in terms of the topological genus of this space. We give a brief economic interpretation of this quantity in this section as well.

 To prove the aforementioned upper bounds for $n$-matching markets $(U_M, U_W)$ it suffices to have just one of the constituent $n$-market profiles $U_M$ or $U_W$ satisfy the polarity property, which we define below.

\begin{definition}[Polarized Market Profiles]
We say that an $n$-market profile $U$ is \textit{polarized} if for all $R \in \SR^n$, and $a,a',x,x' \in [n]$,
\begin{equation} \label{eq:trade}
U[R](a,x') - U[R](a,x) \le -(U[R](a',x) + U[R](a',x'))
\end{equation}
\end{definition}

Informally, this property states that if an agent $a$ strongly prefers some alternative $x'$ to $x$, then every other agent must strongly dislike at least one of $x$ or $x'$. Hence, in colloquially ``polarized'' environments, this property is likely to hold. We claim that one may attach a metric space to the agents' preferences in such markets in a meaningful way. 

\begin{definition}[Metric Space]\label{def:metric}
A metric space $(\X, d_\X)$ is a set $\X$ along with a distance metric $d_\X: \X \times \X \to \R_{\ge 0}$ such that the following properties hold
\begin{align*}
    \forall a,b \in \X, \quad &d_\X(a,b) = 0 \iff a = b && (1) \\
    \forall a,b \in \X, \quad &d_\X(a,b) = d_\X(b,a) && (2) \\
    \forall a,b,c \in \X, \quad &d_\X(a,b) + d_\X(b,c) \ge d_\X(a,c) && (3)
\end{align*}
The metric space is considered finite if $|\X| < \infty$. In this case, we may think of $\X$ as a (non-negatively) weighted undirected graph where $d_\X$ is given by the shortest path metric. The \textit{genus} $g(\X)$ of a finite metric space is defined as the minimum number of handles that must be added to the plane to embed the graph $\X$ without any crossings. An in-depth exposition of these concepts can be found in \cite{mohar2001graphs}. 
\end{definition}

More formally, We claim that any polarized $n$-market profile can be thought of as a collection of maps that represent the agents and alternatives as points in a \textit{generating metric space}. The maps satisfy the property that the utility any agent has for an alternative is given by the negative distance between that agent and the alternative. That is, in a polarized $n$-market profile, we may think of agents as points in space who prefer alternatives that are spatially ``close'' to themselves. This model can be thought of as a generalization of widely-studied Euclidean models of preference.

\begin{definition}[Generating Metric Space]
Let $U$ be an $n$-market profile. We say that a metric space $(\X, d_\X)$ \textit{generates} $U$ if for all $R \in \SR^n$, there exist maps $\alpha_R: [n] \to \X$ and $\beta_R: [n] \to \X$ such that for all $a,x \in [n]$
\begin{equation}
    U[R](a,x) = -d_\X\pa{\alpha_R(a), \beta_R(x)}
\end{equation}
\end{definition}

We now show that polarized $n$-market profiles (and only polarized $n$-market profiles) have a generating metric space $\X$. Further, a finite generating metric space always exists. 

\begin{theorem}\label{thm:genspace}
Let $U$ be a polarized $n$-market profile. There exists a generating metric space $\X$ such that $|\X| \le 2n(n!)^n$. If $U$ is not polarized, no generating metric spaces exist. 
\end{theorem}
\begin{proof}
We first prove that generating metric spaces exist for polarized $n$-market profiles by giving an explicit construction. In our construction, we make a separate metric space for each constituent preference profile $R$. We then ``glue'' all of these metric spaces together by taking the disjoint union to construct the full generating metric space. Each of these separate metric spaces takes the form of a weighted complete bipartite graph.

More formally, for each $R \in \SR^n$, we define $\X_R$ to be a copy of the complete bipartite graph $K_{n,n}$. We let $A := [n]$ and $X := [n]$ denote the two disjoint partitions, and define $\alpha_R: [n] \to A$ and $\beta_R: [n] \to X$ in the canonical way. We let $w_{a,x} := -U[R](a,x) \ge 0$ be the weight of the $(a,x)$ edge, for $a \in A$ and $x \in X$. If the shortest path from any $a$ to $x$ is directly through the $(a,x)$ edge, then the finite metric space given by this graph satisfies the generating metric space property for this preference profile. We show that the polarity property implies that this is the case.

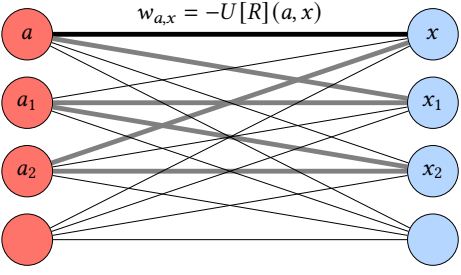
\begin{minipage}{0.49 \textwidth}
\[\scalebox{0.9}{\begin{tikzpicture}
	\begin{pgfonlayer}{nodelayer}
		\node [style=agent, minimum size=.75cm] (0) at (-3, 2) {$a$};
		\node [style=agent, minimum size=.75cm] (1) at (-3, 1) {$a_1$};
		\node [style=agent, minimum size=.75cm] (2) at (-3, 0) {$a_2$};
		\node [style=agent, minimum size=.75cm] (3) at (-3, -1) {};
		\node (8) at (0,2.3) {$w_{a,x} = -U[R](a,x)$};
		\node [style=alt, minimum size=.75cm] (4) at (3, 2) {$x$};
		\node [style=alt, minimum size=.75cm] (5) at (3, 1) {$x_1$};
		\node [style=alt, minimum size=.75cm] (6) at (3, 0) {$x_2$};
		\node [style=alt, minimum size=.75cm] (7) at (3, -1) {};
	\end{pgfonlayer}
	\begin{pgfonlayer}{edgelayer}
		\draw [line width = 2pt] (0) to (4);
		\draw [color=gray, line width = 2pt] (0) to (5);
		\draw [color=gray, line width = 2pt] (5) to (1);
		\draw [color=gray, line width = 2pt] (1) to (6);
		\draw [color=gray, line width = 2pt] (6) to (2);
		\draw [color=gray, line width = 2pt] (2) to (4);
		\draw  (0) to (7);
		\draw  (0) to (6);
		\draw  (1) to (4);
		\draw  (1) to (7);
		\draw  (2) to (5);
		\draw  (2) to (7);
		\draw  (3) to (4);
		\draw  (3) to (5);
		\draw  (3) to (6);
		\draw  (3) to (7);
	\end{pgfonlayer}
\end{tikzpicture}
}\]
\end{minipage}
\begin{minipage}{0.49 \textwidth}
Formally, we show that for any $a \in A$ and $x \in X$, the length of the shortest path $d_{\X_R}(a,x)$ is equal to the edge weight $w_{a,x}$. Suppose for contradiction that the $(a,x)$ edge does not give a shortest path. By construction, the shortest path must then be of the form
\[a \to x_1 \to a_1 \to x_2 \to \dots \to a_k \to x\]
for $a_1, \dots, a_k \ne a$ and $x_1, \dots, x_k \ne x$. Such a path is depicted in gray in the figure on the left. 
\end{minipage}

\vspace{0.5em}

We show by induction on $k$ that for $k > 0$, this path cannot be shorter than the path $a \to x$, thus giving a contradiction. For the case $k = 1$, the length of the path $a \to x_1 \to a_1 \to x$ is given by
\[w_{a,x_1} + w_{a_1,x_1} + w_{a_1,x} = -\pa{U[R](a, x_1) + U[R](a_1, x_1) + U[R](a_1, x)} \ge  -U[R](a, x) = w_{a,x}\]
by polarity thus proving the base case. Next, let $a \to x_1 \to a_1 \to x_2 \to \dots \to a_k \to x$ be the shortest path. By polarity, the path length
\[w_{a,x_1} + w_{a_1,x_1} + w_{a_1,x_2} + w_{a_2, x_2} + \dots + w_{a_k, x_k} + w_{a_k, x} \ge w_{a,x_2} + \dots + w_{a_k, x_k} + w_{a_k, x}\]
which is the length of the path $a \to x_2 \to a_2 \to x_3 \to \dots \to a_k \to x$. By inductive hypothesis, the length of this path is no less than $w_{a,x}$. It thus follows that there exists maps $\alpha_R: [n] \to \X_R$ and $\beta_R: [n] \to \X_R$ such that $U[R](a,x) = -d_{\X_R}(\alpha_R(a), \beta_R(x))$. We now let the $\X$ be given by the disjoint union of metric spaces
\[\X := \coprod_{R \in \SR^n} \X_R\]
whence it follows that $\X$ is a generating metric space. Since $|\X_R| = 2n$ and $|\SR^n| = (n!)^n$, it follows that $|\X| = 2n(n!)^n$. Conversely, if $U$ has a generating metric space $\X$, then by the triangle inequality
\[\dX(\alpha_R(a), \beta_R(x)) \le \dX(\alpha_R(a), \beta_R(x')) + \dX(\alpha_R(a'), \beta_R(x')) + \dX(\alpha_R(a'), \beta_R(x))\]
which implies the polarity condition. 
\end{proof}

The properties of a generating metric space provide information about the structure of the underlying market. If an $n$-market profile has a generating metric space $\X$ of small cardinality, then this indicates that that the agents and alternatives in the market fit into one a few different ``archetypes''. Agents that belong to the same archetype share the same utility function across the alternatives. Similarly, alternatives that belong to the same archetype are assigned the same utility by any agent. 

Another property of interest is the topological \textit{genus} of the metric space. As stated in Definition \ref{def:metric}, this is the minimum number of handles that one must add to the plane in order to embed the graph with no crossings. We attempt to give an economic interpretation for this quantity. Intuitively, the genus places a bound on the number of edges that may exist in the graph. Indeed, a generalization of Euler's formula makes this relationship concrete:
\[V - E + F + 2G = 2\]
which gives a bound $E = \Theta(V + G)$ on the number of edges, where $G$ denotes the genus. Thus, an $n$-market profile has a generating metric space $\X$ with low genus $g(\X)$ if the metric space has few edges. In the case of a finite metric space, if no edge exists between a given pair of points $a$ and $x$, then the shortest path between those points must go through some other point $z$. That is, the triangle inequality must be replaced with an equality over these three points: $\dX(a,x) = \dX(a,z) + \dX(z,x)$. Applying the construction from the proof of Theorem \ref{thm:genspace}, we can see that this occurs when the polarity condition is satisfied exactly for some $R$, $a$, $x$, $a'$, and $x'$:
\[U[R](a,x') - U[R](a,x) = -(U[R](a',x) + U[R](a',x'))\]
That is, an $n$-market profile has a generating metric space of low genus if we can find many pairs of agents $a$ and $a'$, and alternatives $x$ and $x'$, where $a$ holds strong opinions about $x$ relative to $x'$, but $a'$ holds relatively weak opinions about both $x$ and $x'$. The presence of agents in the market that express strength of preference on different scales allows one to use a topologically simple structure to represent the preferences of the agents.

In the next sections, we bound the robustness in terms of the size of a generating metric space, and the probabilistic robustness in terms of the genus. Our bounds show how the underlying complexity in the market must scale with these robustness quantities. Indeed, we see that if agents' preferences may be represented by a simple structure, then the market is not robust.
\section{Upper Bounds on Robustness}\label{sec:robustness}
For any $n$-market profile $U$, in this and subsequent sections, we use the notation $[U]$ to denote any $n$-matching market $(U, U_W)$ or $(U_M, U)$. We now upper bound $\xi_{[U]}$ for any polarized $n$-market profile $U$. There are two key steps in our argument. First, we restrict our attention to a certain family of polarized $n$-market profiles. We consider a mild generalization of Euclidean markets, where the $n$-market profile has a generating metric space that is a finite subset of a complete normed vector space. We show that for this restricted family of polarized $n$-market profiles $U$, the robustness $\xi_{[U]}$ is tightly bounded. 

To obtain an upper bound on $\xi_{[U]}$ for any polarized $n$-market profile $U$, we apply a famous result from the low-distortion embedding literature: that any finite metric space can be embedded in Euclidean space with multiplicative distortion logarithmic in the size of the finite space. This embedding, due to Bourgain \cite{bourgain1985lipschitz}, has an economic interpretation in the context of our framework. Letting $U$ be an arbitrary polarized $n$-market profile, it shows that there exists some Euclidean $n$-market profile $U'$ such that the robustness $[U]$ does not exceed that of $[U']$ by more than a logarithmic factor in the size of the generating metric space. That is, no polarized matching market can be significantly more robust than a Euclidean matching market.

We now proceed with the first step in our argument, and show that matching markets $[U]$ where $U$ has Euclidean generating metric space have low robustness. This result has significance in its own right, as Euclidean models of preference are common assumption in the economics and political science literature \cite{hotelling1990stability, eguia2011foundations, bogomolnaia2007euclidean, davis1972social, enelow1984spatial, anshelevich2017randomized}.

\begin{theorem}\label{thm:banach}
Suppose that $U$ is an $n$-market profile with generating metric space $\X$ such that $\X \subset V$, where $(V, \norm{\cdot})$ is a Banach space. Then, $\xi_{[U]} \le 3$. 
\end{theorem}
\begin{proof}
We first show that the statement holds when $n = 3$. Let $R \in \SR^3$ denote the well-known Condorcet cycle preference profile, given by
\[R: = \begin{array}{|c|c|c|}
\hline
a_1 & a_2 & a_3 \\
\hline
\hline
x_1 & x_2 & x_3 \\
x_2 & x_3 & x_1 \\
x_3 & x_1 & x_2 \\
\hline
\end{array}\]
By Equation \ref{eq:robust}, we have that
\begin{align*} 
\xi_{[U]} &\le \text{\scalebox{1}{$\min\pa{\frac{U[R](a_1, x_3)}{U[R](a_1, x_2)}, \frac{U[R](a_1, x_2)}{U[R](a_1, x_1)}, \frac{U[R](a_2, x_1)}{U[R](a_2, x_3)}, \frac{U[R](a_2, x_3)}{U[R](a_2, x_1)}, \frac{U[R](a_3, x_2)}{U[R](a_3, x_1)}, \frac{U[R](a_3, x_1)}{U[R](a_3, x_3)}}$}}\\
&=\min\pa{\frac{\norm{\alpha_R(a_1) - \beta_R(x_3)}}{\norm{\alpha_R(a_1) - \beta_R(x_2)}}, \cdots , \frac{\norm{\alpha_R(a_3) - \beta_R(x_1)}}{\norm{\alpha_R(a_3) - \beta_R(x_3)}}}\\
&=\min\pa{\frac{\norm{\gamma \alpha_R(a_1) - \gamma \beta_R(x_3)}}{\norm{\gamma \alpha_R(a_1) - \gamma \beta_R(x_2)}}, \cdots , \frac{\norm{\gamma \alpha_R(a_3) - \gamma \beta_R(x_1)}}{\norm{\gamma \alpha_R(a_3) - \gamma \beta_R(x_3)}}}\\
\end{align*}
for any scaling constant $\gamma$. It thus follows that we may restrict our attention to $n$-market profiles $U$ where \[-\min_{a,x \in [3]}(U[R](a,x)) = \max_{a,x \in [3]}(\norm{\alpha_R(a) - \beta_R(x)}) = 1\] We show that this implies that $\xi_{[U]} \le 3$. Let $\alpha_i := \alpha_R(i)$ and $\beta_i := \beta_R(i)$. The figure below depicts the relevant distances in $\X \subset V$. 

\vspace{.5em}

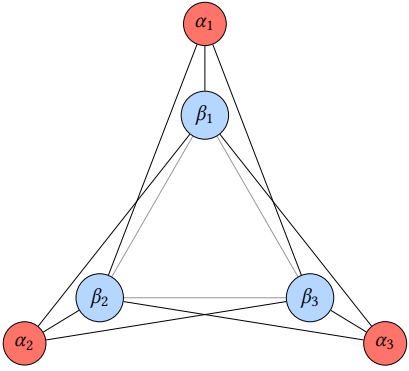
\begin{minipage}{0.4 \textwidth}

\scalebox{0.8}{\begin{tikzpicture}
	\begin{pgfonlayer}{nodelayer}
		\node [style=alt] (0) at (0, 3) {$\beta_1$};
		\node [style=alt] (1) at (-1.75, 0) {$\beta_2$};
		\node [style=alt] (2) at (1.75, 0) {$\beta_3$};
		\node [style=agent] (3) at (-3, -0.75) {$\alpha_2$};
		\node [style=agent] (4) at (0, 4.5) {$\alpha_1$};
		\node [style=agent] (5) at (3, -0.75) {$\alpha_3$};
	\end{pgfonlayer}
	\begin{pgfonlayer}{edgelayer}
		\draw (3) to (1);
		\draw (3) to (2);
		\draw (3) to (0);
		\draw (5) to (2);
		\draw (1) to (5);
		\draw (0) to (5);
		\draw (4) to (0);
		\draw (4) to (1);
		\draw (4) to (2);
		\draw [style=grey] (0) to (2);
		\draw [style=grey] (0) to (1);
		\draw [style=grey] (1) to (2);
	\end{pgfonlayer}
\end{tikzpicture}
}

\end{minipage}
\begin{minipage}{0.59 \textwidth}
Suppose that 
\[\norm{\alpha_1 - \beta_3} = \max_{a,x \in [3]}(\norm{\alpha_R(a) - \beta_R(x)}) = 1\]
If $[U]$ is $C$-robust in addition to the above, the following must hold
\[\norm{\alpha_1-\beta_3} = 1 \quad \norm{\alpha_1 - \beta_2} \le \frac{1}{C} \quad \norm{\alpha_1 - \beta_1} \le \frac{1}{C^2}\]
\[\norm{\alpha_2-\beta_1} \le 1 \quad \norm{\alpha_2 - \beta_3} \le \frac{1}{C} \quad \norm{\alpha_2 - \beta_3} \le \frac{1}{C^2}\]
\[\norm{\alpha_3-\beta_2} \le 1 \quad \norm{\alpha_3 - \beta_1} \le \frac{1}{C} \quad \norm{\alpha_3 - \beta_3} \le \frac{1}{C^2}\]
\end{minipage}

\vspace{.5em}

By the triangle inequality, we then have that
\begin{align*}
    1 - \frac1{C^2} \le \norm{\alpha_1 - \beta_3} - \norm{\alpha_1 - \beta_1} \le \norm{\beta_1 - \beta_3} &\le \norm{\beta_1 - \beta_2} + \norm{\beta_2 - \beta_3}\\
    &\le \norm{\alpha_1 - \beta_1} + \norm{\alpha_1 - \beta_2} + \norm{\alpha_2 - \beta_2} + \norm{\alpha_2 - \beta_3}\\
    &\le 2\pa{\frac{1}{C^2} + \frac1C}
\end{align*}
whence it follows that $1 \le \frac{2}{C} + \frac{3}{C^2} \implies C \le 3$ as desired. This result extends to any $n$ as by Equation \ref{eq:robust},
\begin{align*}
    \xi_{[U]} \le \min_{R \in \SR^n} \pa{ \min_{a,x,x' \in [n] \mid xR_ax'}\pa{\frac{U[R](a,x')}{U[R](a,x)}}} &\le \min_{R \in \SR^n} \pa{ \min_{a,x,x' \in [3] \mid xR_ax'}\pa{\frac{U[R](a,x')}{U[R](a,x)}}} \\
    &=\min_{R \in \SR^3} \pa{ \min_{a,x,x' \in [3] \mid xR_ax'}\pa{\frac{U[R](a,x')}{U[R](a,x)}}}\le 3\\
\end{align*}
by the above. 
\end{proof}

By applying Proposition \ref{prop:adm}, we obtain a strong negative result with economic significance.
\begin{corollary}
Suppose that $\set{[U_n]}_{n=1}^\infty$ is a collection of matching markets, where $U_n$ is an $n$-market profile whose generating metric space $\X_n \subset V$, where $(V, \norm{\cdot})$ is a Banach space. Then, if $H$ is admissible for the collection, then $H(n) = O(1)$.
\end{corollary}
That is, if we hope to obtain feasible communication requirements in large markets with Euclidean preferences, then the hardness of preference learning must remain \textit{constant} as the size of the market increases. In practice, this means that agents must be able to interview many candidates in the same time as it would take for them to interview a few. As this condition may not hold in a real-world market, this result indicates that stable matching mechanisms will not function properly in large markets where agents have Euclidean preferences. 

We now proceed with the second step in our argument, and use Bourgain's embedding \cite{bourgain1985lipschitz} to extend this result to any $[U]$ where $U$ is a polarized $n$-market profile.
\begin{theorem}[Bourgain]\label{thm:bourgain}
For any finite metric space $\X$, there exists an injective map $T: \X \to \ell_2$ such that for any $a,b \in \X$,
\[1 \le \frac{\norm{T(a) - T(b)}_2}{\dX(a,b)} \le O(\log |\X|)\]
\end{theorem}
\begin{theorem}\label{thm:detbound}
Let $U$ be a polarized $n$-market profile. Then, $\xi_{[U]} = O(\log |\X|)$ where $\X$ is a generating metric space for $U$ of minimum cardinality. Notably, $\xi_{[U]} = O(n^2 \log n)$.
\end{theorem}
\begin{proof}
Let $T: \X \to \ell_2$ be as given in Theorem \ref{thm:bourgain}. Let $U'$ be the $n$-market profile given by
\[U'[R](a,x) := -\norm{(T \circ \alpha_R)(a) - (T \circ \beta_R)(x)}_2\]
where $\alpha_R: [n] \to \X$ and $\beta_R: [n] \to \X$ are the maps into the generating metric space. By Theorem \ref{thm:banach}, we must have that $\xi_{[U']} \le 3$, since the maps $T \circ \alpha_R$ and $T \circ \beta_R$ yield a generating metric space for $U'$ that is a subset of $\ell_2$. It then follows that for any $R \in \SR^n$ and $a,x,x' \in [n]$ such that $xR_ax'$
\[\frac{U[R](a,x')}{U[R](a,x)} \le O(\log |\X|)\frac{U'[R](a,x')}{U'[R](a,x)}\]
whence it follows by Equation \ref{eq:robust} that
\[\xi_{[U]} \le O(\log |\X|)\xi_{[U']} \implies \xi_{[U]} = O(\log |\X|)\]
as desired. By Theorem \ref{thm:genspace}, there exists a generating metric space $\X$ for $U$ of size $2n(n!)^n$. Thus, 
\[\xi_{[U]} = O(\log |\X|) = O(\log(2n(n!)^n)) = O(n^2\log n)\]
\end{proof}

We again apply Proposition \ref{prop:adm} to get a corollary
\begin{corollary}\label{cor:rob}
Suppose that $\set{[U_n]}_{n=1}^\infty$ is a collection of matching markets, where $U_n$ is a polarized $n$-market profile. Then, if $H$ is admissible for the collection, then $H(n) = O(n^2 \log n)$. Conversely, given an $H(n)$ that is admissible for the collection, we must have that $|\X_n| = \Omega(e^{H(n)})$, where $\X_n$ denotes the generating metric space for $U_n$.
\end{corollary}

In other words, if the task of learning one's preferences becomes more difficult by a factor that asymptotically exceeds $n^2 \log n$ as the number of agents $n$ in the market increases, then it is unlikely that stable matching mechanisms will function properly in large markets. Further, if we would like guarantees on the communication requirements for large markets given that the hardness of preference learning scales with $H(n)$ as the size of the matching market $n$ increases, then the corresponding spatial structure of the market must increase in size exponential to $H(n)$. This result demonstrates that the underlying structure in the market must quickly become increasingly complex for stable matching mechanisms to function properly in large markets. We show a similar statement in the next section as it relates to probabilistic admissibility.  
\section{Upper Bounds on Probabilistic Robustness}\label{sec:probounds}
The main theoretical tool we used to generalize this result to the general class of polarized matching markets was Bourgain's embedding -- a result primarily used in the theoretical computer science literature to speed up the running time of algorithms that work with geometric inputs. In this section, we apply a similar approach to bound the probabilistic robustness $\xi_{[U]}^P$ for any polarized $U$ in terms of the topological complexity of its generating metric space. 

As before, our argument has two key steps. However, the first step differs slightly from the previous section. Rather than showing that a class of polarized markets has tightly bounded probabilistic robustness, we instead show that there exists some preference profile $R$, such that if one were to represent $R$ in a generating metric space, then the metric space must not have genus zero. Our argument is an explicit one: we construct such a preference profile $R$, and show that any arbitrary generating metric space that represents this profile must contain a $K_{3,3}$ minor. By Kuratowski's theorem, it then follows that the finite metric space is nonplanar, and has genus at least $1$. 

In the second step we again apply a result from the low-distortion embedding literature: this time we use Sidiropolous's probabilistic embedding of finite metric spaces into planar graphs \cite{sidiropoulos2010optimal}. This embedding allows one to map any metric space into a random planar graph such that in expectation, distances are preserved up to a multiplicative factor. Sidiropolous shows that this factor is logarithmic in the genus of the original metric space. 

By applying Sidiropolous' embedding to a generating metric space, we obtain a second generating metric space for another market. Thus, the embedding can be thought of as a probabilistic perturbation that transforms the original utilities to the utilities of the transformed market. Thinking of the embedding in this way, we can see that if a polarized market has true ordinal preferences given by $R$ (where $R$ is the aforementioned preference profile), then upon applying the embedding, the preference profile must deviate to some other preference profile $R'$. This occurs as we may not represent $R$ in a planar metric space -- yet by applying the embedding we transform the generating metric space into a planar metric space with probability $1$. Thus, we obtain a bound on the probabilistic robustness.

We make this intuition concrete and proceed with the first step. Recall that we aim to construct a preference profile $R$ such that any metric space that represents the profile has a $K_{3,3}$ minor. To show the existence of such a minor, we make use of the following Lemma.

\begin{lemma}\label{lem:join}
Let $U$ be an $n$-market profile with generating metric space $\X$. If some shortest paths $P$ from $\alpha_R(a)$ to $\beta_R(x)$ and $P'$ from $\alpha_R(a')$ to $\beta_R(x')$ intersect at any vertex $v \in \X$, then $xR_{a}x' \iff xR_{a'}x'$
\end{lemma}
\begin{proof}
For paths $P \subset \X$ of a finite metric space $\X$ containing vertices $u,v \in \X$, we use the notation $P(u,v) \subset P$ to denote the subpath of $P$ from $u$ to $v$. As $P$ and $P'$ are shortest paths, we must have that $P(v,\beta_R(x))$ and $P'(v,\beta_R(x'))$ are shortest paths. It then follows that $P(\alpha_R(a), v) \cup P'(v, \beta_R(x'))$ and $P'(\alpha_R(a'), v) \cup P(v, \beta_R(x))$
are also shortest paths. Thus, we have that
\[\scalebox{0.9}{\text{$\dX(\alpha_R(a),\beta_R(x')) = \dX(\alpha_R(a), v) + \dX(v, \beta_R(x')) \qquad \dX(\alpha_R(a),\beta_R(x)) = \dX(\alpha_R(a'), v) + \dX(v, \beta_R(x))$}}\]
It then follows by the definition of a generating metric space that
\begin{align*}
xR_{a}x' \iff d_\X(\alpha_R(a), \beta_R(x)) < \dX(\alpha_R(a), \beta_R(x')) &\iff d_\X(v, \beta_R(x)) < \dX(v, \beta_R(x'))\\
&\iff d_\X(\alpha_R(a'), \beta_R(x)) < \dX(\alpha_R(a'), \beta_R(x'))\\
&\iff xR_{a'}x'
\end{align*}
\end{proof}

Lemma \ref{lem:join} allows us to prove the existence of edges in a generating metric space given a preference profile. We now use this to show the existence of a ``topologically complex'' preference profile $R$. 

\begin{lemma}\label{lem:planar}
Let $U$ be an $n$-market profile for $n \ge 9$. There exists a preference profile $R \in \SR^n$ such that if maps $\alpha_R: [n] \to \X$ and $\beta_R: [n] \to \X$ exist such that $U[R](a,x) = -\dX(\alpha_R(a), \beta_R(x))$, then $\X$ has genus greater than zero.
\end{lemma}
\begin{proof}
We show that $\X$ has a $K_{3,3}$ minor, whence it follows that $\X$ may not be genus $0$. Let $R$ be any preference profile where

\begin{minipage}{0.44 \textwidth}
\[\scalebox{0.8}{$R: = \begin{array}{|c|c|c|c|c|c|c|c|c|}
\hline
a_1 & a_2 & a_3 & a_4 & a_5 & a_6 & a_7 & a_8 & a_9 \\
\hline
\hline
x_1 & x_2 & x_3 & x_4 & x_5 & x_6 & x_7 & x_8 & x_9 \\
x_4 & x_5 & x_6 & x_2 & x_3 & x_1 & x_4 & x_5 & x_6 \\
* & * & * & * & * & * & x_3 & x_1 & x_2 \\
* & * & * & * & * & * & * & * & * \\
\vdots & \vdots &\vdots &\vdots &\vdots &\vdots &\vdots &\vdots &\vdots\\
* & * & * & * & * & * & * & * & * \\
\hline
\end{array}$}\]
\end{minipage}
\begin{minipage}{0.54 \textwidth}
For $a,x \in [n]$, we use the notation $\Ps(a,x) \subset 2^\X$ to denote the collection of shortest paths from $\alpha_R(a)$ to $\beta_R(x)$. Notice that under the profile $R$, each agent $a_i$ prefers the alternative $x_i$ the most. It follows by Lemma \ref{lem:join} that the $\Ps(a_i,x_i)$ are disjoint subsets of $\X$. Contracting each of these subsets $\Ps(a_i,x_i)$ into a single node, we now show that this contraction of the graph $\X$ has a $K_{3,3}$ minor. 
\end{minipage}

\vspace{0.5em}

Below is a diagram of the contraction. To prove the existence of the minor, we show that (i) all of the black paths exist and do not intersect each other unless they share a source/destination, and (ii) at least one red path, one green path, and one blue path exists and does not intersect any other path unless they share a source/destination. It then follows that by contracting along the dashed black paths, we obtain a $K_{3,3}$ minor. It is acceptable for edges that share a source or destination to intersect, as we may still obtain a minor by contracting the path between the intersection point and the source/destination. 

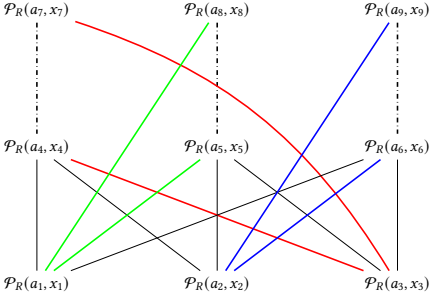
\begin{minipage}{0.5 \textwidth}
\[\scalebox{0.6}{\begin{tikzpicture}
	\begin{pgfonlayer}{nodelayer}
		\node  (0) at (0, 0) {$\Ps(a_5, x_5)$};
		\node  (1) at (-4, 0) {$\Ps(a_4,x_4)$};
		\node  (2) at (4, 0) {$\Ps(a_6,x_6)$};
		\node  (3) at (-4, 3) {$\Ps(a_7,x_7)$};
		\node  (4) at (0, 3) {$\Ps(a_8,x_8)$};
		\node  (5) at (4, 3) {$\Ps(a_9,x_9)$};
		\node  (6) at (-4, -3) {$\Ps(a_1,x_1)$};
		\node  (7) at (0, -3) {$\Ps(a_2,x_2)$};
		\node  (8) at (4, -3) {$\Ps(a_3,x_3)$};
	\end{pgfonlayer}
	\begin{pgfonlayer}{edgelayer}
		\draw (1) to (6);
		\draw (0) to (7);
		\draw (2) to (8);
		\draw (1) to (7);
		\draw (0) to (8);
		\draw [dash dot, line width = 1pt] (3) to (1);
		\draw [dash dot, line width = 1pt] (4) to (0);
		\draw [dash dot, line width = 1pt] (5) to (2);
		\draw [bend left=20, color=red, line width = 1pt] (3) to (8); 
		\draw [color=red, line width = 1pt] (1) to (8); 
		\draw [color=green, line width = 1pt] (0) to (6); 
		\draw [color=green, line width = 1pt] (4) to (6); 
		\draw (2) to (6);
		\draw [color=blue, line width = 1pt] (2) to (7);  
		\draw [color=blue, line width = 1pt] (5) to (7); 
	\end{pgfonlayer}
\end{tikzpicture}}\]
\end{minipage}
\begin{minipage}{0.49\textwidth}
We first show the existence of the black paths. Observe that we have drawn a black path between the contracted nodes $\Ps(a_i, x_i)$ and $\Ps(a_j, x_j)$ if $x_j$ is the second ranked alternative by the agent $a_i$ under $R$. Since $\X$ is a metric space, the distance from $a_i \in \Ps(a_i, x_i)$ to $x_j \in \Ps(a_j, x_j)$ must be finite, so a black path must exist. We consider such paths that fall within the collection of \textit{shortest} paths $\Ps(a_i, x_j)$. 
\end{minipage}

\vspace{0.5em}

Having shown existence, we show that any of these shortest paths do not intersect any contracted noted $\Ps(a_k, x_k)$ where $k \not \in \set{i,j}$. To see this, we apply Lemma \ref{lem:join}. Recall that $x_j$ is agent $a_i$'s second ranked preference under $R$. We must therefore have that $\Ps(a_i, x_j) \cap \Ps(a_k, x_k) = \emptyset$  since $x_jR_{a_i}x_k$ but $x_kR_{a_k}x_j$. We can similarly see by Lemma \ref{lem:join} that any shortest path in $\Ps(a_i, x_j)$ does not intersect any shortest path in $\Ps(a_k, x_l)$ for $i \ne k$ and $j \ne l$, since $x_jR_{a_i}x_\ell$ but $x_lR_{a_k}x_j$. Thus, we may select black paths such that none of the paths intersect each other or a contracted node. 

Finally, we now show that at least one of red paths exist, and does not intersect any other path or contracted node. By symmetry, we obtain a similar guarantee for the green and blue paths. As before, we consider red paths from the collection of shortest paths $\Ps(a_7, x_3)$. Observe that $x_3$ is agent $a_7$'s third ranked preference under $R$, and that $x_4$ is agent $a_7$'s second ranked preference under $R$. It then follows by Lemma \ref{lem:join} that for all $i \ne 4$, the collection of shortest paths $\Ps(a_7, x_3) \cap \Ps(a_i, x_i) =\emptyset$ have no intersection since $x_3R_{a_7}x_i$, but $x_iR_{a_i}x_7$. Thus, none of our candidate red paths may intersect with any contracted node -- with the exception of the node $\Ps(a_4, x_4)$. 

We split into three cases. We first consider the case where at least one path $P$ in $\Ps(a_7, x_3)$ does not intersect with any paths in $\Ps(a_4, x_4)$ or $\Ps(a_4, x_2)$ (recall that $x_2$ is $a_4$'s second ranked preference). In this case, we may verify that by Lemma \ref{lem:join}, $P \cap \Ps(a_i, x_j) = \emptyset$ where $i \ne 7$, $j \ne 3$, and $x_j$ is $a_i$'s second ranked preference. That is, a ``top'' red path exists and does not intersect any of the contracted nodes or other paths as desired. Next, we consider the case where at least one path $P$ in $\Ps(a_7, x_3)$ intersects with just $\Ps(a_4, x_4)$. In this case, we get a ``bottom'' path that travels from the contracted node $\Ps(a_4, x_4)$ to the node $\Ps(a_3, x_3)$. By the same argument as in the first case, we still have that $P \cap \Ps(a_i, x_j) = \emptyset$ where $i \ne 7$, $j \ne 3$. Thus, this bottom path does not intersect any of the contracted nodes or other paths. Finally, we consider the case where all shortest paths in $\Ps(a_7, x_3)$ intersect $\Ps(a_4, x_2)$. Let $P$ be any such path. We extend the contracted node $\Ps(a_4, x_4)$ by contracting along the path from the node to the intersection point. We have now reduced this case to the previous one. By applying the same argument with the path $P$, we obtain the desired result.
\end{proof}

We now proceed with the second step in our argument, and use Sidiropolous's embedding \cite{sidiropoulos2010optimal} to extend this result as we did in Section \ref{sec:robustness}.
\begin{theorem}[Sidiropoulos]\label{thm:sid}
For any finite metric space $\X$, there exists an injective map $T: \X \to \X'$ where $\X'$ is a random planar graph such that for any $a,b \in \X$, 
\[1 \le \frac{\ex{d_{\X'}(a,b)}}{\dX(a,b)} \le O(\log g(\X))\]
\end{theorem}
\begin{theorem}
Let $U$ be an $n$-market profile with generating metric space $\X$. Then, $\xi_{[U]}^P = O(\log g(\X))$ and $\xi_{[U]} = O(n^2 \log g(\X))$.
\end{theorem}
\begin{proof}
Without loss of generality, suppose $[U] = (U, U_W)$. It suffices to construct a joint distribution over preference profiles $R^M, R^W \in \SR^n$ and probabilistic $O(\log g(\X))$-perturbations $\delta_M$ and $\delta_W$ such that \[\pr{\Phi(R^M, R^W) = \Phi\pa{\SR(\delta_M U[R^M]), \SR(\delta_W U[R^W])}} = 0\]
Let $T$ be as given in \ref{thm:sid}. For $a,x \in [n]$, We choose 
\[\delta_M(a,x) := \frac{d_{T(\X)}(\alpha_R(a),\beta_R(x))}{\dX(\alpha_R(a),\beta_R(x))}\]
By Theorem \ref{thm:sid}, $\delta_M$ is a probabilistic $O(\log g(\X))$-perturbation. We let $\delta_W$ be the trivial $1$-perturbation. Let $R^M \in \SR^n$ be as given in Lemma \ref{lem:planar}. By Lemma \ref{lem:planar}, we must have that with probability $1$, 
\[\SR(\delta_MU[R^M]) = R\pa{-\delta_M \pa{\dX \circ (\alpha_R, \beta_R)}} = R\pa{-\pa{d_{T(\X)} \circ (\alpha_R, \beta_R)}} \ne R^M\]
since $T(\X)$ is planar with probability $1$. By Lemma \ref{lem:ne}, there exists a preference profile $R'$ such that $\Phi(\SR(\delta_MU[R^M]), R') \ne \Phi(R^M, R')$. Choosing $R^W = R'$, we have that
\[\pr{\Phi(R^M, R^W) = \Phi\pa{\SR(\delta_M U[R^M]), \SR(\delta_W U[R^W])}} = 0\]
as desired. By Theorem \ref{thm:detbound}, we have that $\xi_{[U]} = O(n^2\xi_{[U]}^P)$ whence the second part of the theorem follows.
\end{proof}

By Proposition \ref{prop:adm}, we obtain a corollary
\begin{corollary}
Suppose that $\set{[U_n]}_{n=1}^\infty$ is a collection of matching markets, where $U_n$ is a polarized $n$-market profile with generating metric space $\X_n$. If $H^P$ is probabilistically admissible for the collection, then the genus $g(\X_n) = \Omega(e^{H^P(n)})$. Similarly, if $H$ is admissible for the collection, then the genus $g(\X_n) = \Omega(e^{H(n)/n^2})$.
\end{corollary}

This statement can be seen as a direct extension of Corollary \ref{cor:rob} for probabilistic admissibility. If we would like a probabilistic guarantee on the behavior of stable matching mechanisms in large markets where preference learning becomes harder on the order of $H^P(n)$, then the underlying topological complexity of the market must increase exponentially in $H^P(n)$. Applying the economic intuition about genus described in Section \ref{sec:polarization}, this implies more informally that for stable matching mechanisms to function in large markets, the agents must quickly converge to having preferences that are on the same scale as each other. 
\section{Commentary and Open Problems} \label{sec:conclusion}

In this paper, we introduced a framework for understanding how communication requirements for stable matching depend on some higher-order structure in the market. In a broad family of markets where preference has a spatial interpretation, we showed that the underlying complexity of the market must increase exponentially in the difficulty of preference learning for stable matching mechanisms to function properly in a finite amount of time. We considered two different measures of complexity: the size and genus of the underlying metric space. We gave an interpretation of each of these quantities. The size of the space can be thought of as a proxy for the number of agent and alternative archetypes. A generalization of Euler's formula shows that the genus is inversely related to the extent to which agents express strength of preference on different scales. By proving that a generating metric space of a certain size always exists, we also obtained a bound depending only on the size of the market. A summary of these results are given in Table \ref{oresults}.

\begin{table}[h!] 
\centering
\resizebox{0.8\columnwidth}{!}{\begin{minipage}{\textwidth}\begin{center}

    \begin{tabular}{|c||c|c|c|}
        \hline
         & $|\X|$ & $g(\X)$ & $n$  \\
        \hline
        \hline
        $T_{(U_M, U_W)}$ & $\Omega\pa{D^{-1}\pa{\frac{H(n)}{\log |\X|}}}$ & $\Omega\pa{D^{-1}\pa{\frac{H(n)}{n^2\log g(\X)}}}$ & $\Omega\pa{D^{-1}\pa{\frac{H(n)}{n^2\log n}}}$ \\
        \hline
        $T^P_{(U_M, U_W)}$ & $\Omega\pa{D^{-1}\pa{\frac{H(n)}{\log |\X|}}}$ & $\Omega\pa{D^{-1}\pa{\frac{H(n)}{\log g(\X)}}}$ & $\Omega\pa{D^{-1}\pa{\frac{H(n)}{n^2\log n}}}$ \\
        \hline
        $H(n)$ & $O(\log |\X_n|)$ & $O(n^2\log g(\X_n))$ & $O(n^2\log n)$ \\
        \hline 
        $H^P(n)$ & $O(\log |\X_n|)$ & $O(\log g(\X_n))$ & $O(n^2\log n)$ \\
        \hline
    \end{tabular}\end{center}\caption{\label{oresults}Lower bounds on the (probabilistic) communication requirements for stable matching and upper bounds on admissible $H$ and probabilistically admissible $H^P$.}\end{minipage}
    }
\end{table}
\vspace{-1.5em}
Our results leave several open problems. In our analysis, we used results that proved the existence of low-distortion embeddings to show upper bounds on the robustness of matching markets. As it turns out, both of the low-distortion embedding existence results that we used (Bourgain's embedding and Sidiroplous's embedding) are optimal \cite{abraham2006advances, sidiropoulos2010optimal} -- one may not prove a tighter bound on the robustness using the same technique. More interestingly, the optimality of both of these embeddings is proven by the same counter-example: constant-degree vertex expander graphs. That is, constant-degree expanders are among the ``hardest'' metric spaces to distort into simpler structures. Thus, we conjecture that one may generate optimal \textit{lower bounds} on robustness by considering markets with a generating metric space given by a constant-degree expander. A proof of this conjecture would have tangible economic significance, as it would provide a construction for a stable matching market that has minimal communication requirements. Such a construction may be useful in the field of mechanism design. 

Another open problem relates the definition of our model in Section \ref{subsec:model}. We left the definition of the decay $D$ nearly arbitrary in the setup of our model, and instead used its properties to bound admissible values of $H$ given the robustness in Proposition \ref{prop:adm}. However, it might be possible to develop an empirical study that measures some analogue of $D$ and $H$ explicitly in a given market. Given this data, our bounds on $T_{(U_M, U_W)}$ and $T^P_{(U_M, U_W)}$ may be used directly by market designers to understand whether the stable matching mechanism is failing due to impractically high communication requirements, or if the market may in fact be cleared earlier due low communication requirements. 

\newpage

\bibliographystyle{ACM-Reference-Format}
\bibliography{refs}

\newpage
\appendix
\section{Proof of Theorem \ref{thm:ub}}\label{proof:ub}

\begin{proof}
For $R \in \SR^n$ and $a,i \in [n]$, we use the notation $X(R, a, i)$ to denote the alternative $x \in [n]$ that is ranked at position $i$ by agent $a$ under the preference profile $R$. We first show that $(U_M, U_W)$ is probabilistically $C$-robust. Since $(U_M, U_W)$ is a $\pa{2n(n-1)(C-1) + 1}$-robust matching market, by Theorem \ref{thm:c}, we have that for all preference profiles $R \in \SR^n$,and $m,w,w' \in [n]$, \[wR_mw' \implies \pa{2n(n-1)(C-1) + 1}U_M[R](m,w) > U_M[R](m,w')\] and similarly for $U_W[R]$. For any joint distribution over the preference profiles and probabilistic $C$-perturbations $R^M$, $R^W$, $\delta_M$, and $\delta_W$, we show that the probability that the ordinal rankings are unchanged by the perturbation\[\pr{\SR(\delta_M U_M[R])= R^M \cap \SR(\delta_W U_W[R])= R^M} > 0\] whence it follows that $(U_M, U_W)$ is probabilistically $C$-robust. Let $\delta_M$ and $\delta_W$ be arbitrary probabilistic $C$-perturbations. We define a collection of $2n(n-1)$ random variables. For $a \in [2n]$ and $i \in [n-1]$, we define
\[D_{(a, i)} = \begin{cases}\delta_M\pa{a, X(R^M, a, i)}  & a \le n \\ \delta_W\pa{a - n, X(R^W, a-n, i)} & a > n \end{cases}\]
which denote the multiplicative distortions induced by $\delta_M$ and $\delta_W$. 
It suffices to show that \[\pr{\max_{(a, i) \in [2n] \times [n-1]} D_i\le 2n(n-1)(C-1) + 1} > 0\]
as by $\pa{2n(n-1)(C-1) + 1}$-robustness, this implies that \[\pr{\SR(\delta_M U_M[R])= R^M \cap \SR(\delta_W U_W[R])= R^M} > 0\]
The random variables $\delta_M(m, X(R^M, m, n))$ and $\delta_W(w, X(R^W, w, n))$ need not be included as agents may not overestimate their utility under a distortion, and $w_n^m$ and $m_n^w$ are already ranked last. Note that for all $i$, $\ex{D_{(a, i)}} \le C$ and $D_{(a, i)} \ge 1$. It then follows that
\[\ex{\sum_{(a, i) \in [2n] \times [n-1]} D_{(a,i)}} = \sum_{(a, i) \in [2n] \times [n-1]} \ex{D_{(a,i)}} \le \sum_{(a, i) \in [2n] \times [n-1]} C = 2n(n-1)C\]
whence it follows that
\[\pr{\sum_{(a,i) \in [2n] \times [n-1]} D_{(a, i)} \le 2n(n-1)C} > 0\]
Since $D_i \ge 1$, we have that
\begin{align*}\sum_{(a,i) \in [2n] \times [n-1]} D_{(a,i)} \le 2n(n-1)C \implies \max_{(a,i) \in [2n] \times [n-1]} D_{(a,i)} &\le 2n(n-1)C - (2n(n-1)-1) \\&= 2n(n-1)(C-1) + 1\end{align*}
whence the desired result follows.

Next, we show that for any $\epsilon > 0$, there exists an $n$-matching market $(U_M, U_W)$ that is $\pa{2n(n-1)(C-1) + 1}$-robust, but not probabilistically $(1+\epsilon)C$-robust. Let $(U_M, U_W)$ be such that for every $R \in \SR^n$, $a \in [n]$, and $i \in [n-1]$
\[\frac{U_M[R](a, X(R^M, a, i+1))}{U_M[R](a, X(R^M, a, i))} = \frac{U_W[R](a, X(R^W, a, i+1))}{U_W[R](a, X(R^W, a, i))} = 2n(n-1)\pa{\pa{1 + \frac{\epsilon}{2}}C-1} + 1\]
By Theorem \ref{thm:c}, $(U_M, U_W)$ is $\pa{2n(n-1)(C-1) + 1}$-robust. We show that $(U_M, U_W)$ is not probabilistically $(1+\epsilon)C$ robust by defining an appropriate joint distribution over preference profiles $R^M, R^W \in \SR^n$ and probabilistic $C$-perturbations $\delta_M$ and $\delta_W$. We define a collection of $2n^2$ random variables. We first select some $(a^*,i^*) \in [2n]\times[n-1]$ uniformly at random. We then define, for $(a,i) \in [n] \times [n-1]$
\[\Delta^M_{(a,i)} := \begin{cases}2n(n-1)\pa{(1+\epsilon)C - 1} + 1 & a = a^*, i = i^*\\1&\elsec\end{cases}\]
\[\Delta^W_{(a,i)} := \begin{cases}2n(n-1)\pa{(1+\epsilon)C - 1} + 1 & a-n = a^*, i = i^*\\1&\elsec\end{cases}\]
We then define, for all $a \in [n]$, $\Delta^M_{(a,n)} = 1$. Notice that either all of the $\Delta^M$ variables will be equal to $1$, and one $\Delta^W$ variable will not, or all $\Delta^W$ variables will be equal to $1$ and one $\Delta^M$ variable will not. Suppose without loss of generality that all of the $\Delta^W$ variables are $1$. We pick $R^M$ uniformly at random, and let
\[\delta_M(m, X(R^M, m, i)) = \Delta_{(m,i)}\]
for all $m \in [n]$. We have that $\delta_M$ is a probabilistic $(1+\epsilon)C$-perturbation, as $\Delta_{(m,i)} \ge 1$ with probability $1$, and for all $m \in [n]$ and $i \in [n-1]$,
\begin{align*}
    \ex{\Delta_{(m,i)}} &= \frac{2n(n-1)\pa{(1+\epsilon)C - 1} + 1}{2n(n-1)} + 1 - \frac{1}{2n(n-1)}\\
    &= (1+\epsilon)C
\end{align*}
Next, notice that by the construction of $(U_M, U_W)$, we must have that $\SR(\delta_MU_M[R^M]) \ne R^M$. By Lemma \ref{lem:ne}, there exists some $R'$ such that $\Phi(\SR(\delta_MU_M[R^M]), R') \ne \Phi(R^M, R')$. We pick $R^W = R'$. Lastly, we pick $\delta_W(w,m) = 1$ for all $w,m \in [n]$. Under this joint distribution, it then follows that
\[\pr{\Phi(R^M, R^W) = \Phi(\SR(\delta_MU_M[R^M]), \SR(\delta_WU_W[R^W]))} = 0\]
whence we have that $(U_M, U_W)$ is not probabilistically $(1+\epsilon)C$-robust.
\end{proof}

\end{document}